\documentclass[aps,reprint,nofootinbib]{revtex4-1}

\usepackage{amsmath}
\usepackage{amssymb}
\usepackage{amsfonts}
\usepackage{amsthm}
\usepackage{bm}
\usepackage{graphicx}
\usepackage{color}

\newtheorem{lemma}{Lemma}
\newcommand{\Ref}[1]{Ref.~\cite{#1}}

\newcommand{\SEC}[2]{\section{\label{sec:#1}#2}}
\newcommand{\Sec}[1]{Sec.~\ref{sec:#1}}
\newcommand{\FIG}[2]{\caption{\label{fig:#1}#2}}
\newcommand{\Fig}[1]{Fig.~\ref{fig:#1}}
\newcommand{\EQ}[1]{\label{eq:#1}}
\newcommand{\Eq}[1]{Eq.~(\ref{eq:#1})}

\newcommand{\prm}[1]{#1^\prime}
\newcommand{\prmprm}[1]{#1^{\prime\prime}}
\newcommand{\pnt}[1]{\bm{#1}}
\newcommand{\avg}[1]{\langle#1\rangle}

\begin{document}

\title{Second-order fluctuation theory and time autocorrelation function for currents}
\date{\today}
\author{Roman Belousov}\email{belousov.roman@gmail.com}
\affiliation{The Rockefeller University, New York 10065, USA}
\author{E.G.D. Cohen}\email{egdc@mail.rockefeller.edu}
\affiliation{The Rockefeller University, New York 10065, USA}
\affiliation{Department of Physics and Astronomy, The University of Iowa, Iowa
City, Iowa 52242, USA}
\begin{abstract}
By using recent developments for the Langevin dynamics of spatially asymmetric
systems, we routinely generalize the Onsager-Machlup fluctuation theory of the
second order in time. In this form, it becomes applicable to fluctuating variables,
including hydrodynamic currents, in equilibrium as well as nonequilibrium steady
states. From the solution of the obtained stochastic equations we derive an analytical
expression for the time autocorrelation function of a general fluctuating quantity.
This theoretical result is then tested in a study of a shear flow by molecular
dynamics simulations. The proposed form of the time autocorrelation function yields
an excellent fit to our computational data for both equilibrium and nonequilibrium
steady states. Unlike the analogous result of the first-order Onsager-Machlup theory,
our expression correctly describes the short-time correlations. Its utility is
demonstrated in an application of the Green-Kubo formula for the transport coefficient.
Curiously, the normalized time autocorrelation function for the shear flow, which
only depends on the deterministic part of the fluctuation dynamics, appears independent
of the external shear force in the linear nonequilibrium regime.
\end{abstract}
\maketitle

\SEC{intro}{Introduction}

The stochastic theory of fluctuations for physical systems in equilibrium, due to
Onsager and Machlup, was originally presented in two forms, Refs.~\cite{OM1953} and \cite{MO1953},
respectively, which were published together in the same journal issue. One of these
papers \cite{OM1953} describes a model based on the Langevin differential equation
of the first order in time, while the other \cite{MO1953} was concerned with its
extension to the second order in time. This latter model received relatively little
attention compared to the first-order theory, which was much broader disseminated
and is now included in classic and modern text books on Statistical Physics, e.g.,
Refs.~\cite[Chapter XII]{LLSt} or \cite[Chapter 2]{Attard}.

The first-order fluctuation theory was generalized also to nonequilibrium steady
states, e.g., Refs.~\cite{KSSH2015I,KSSH2015II,MorgadoQ2016,PRE2016II}. The principle
argument, on which Onsager and Machlup mainly relied in their papers, was the time
reversibility. The recent developments \cite{PRE2016,PRE2016II,Bonella2014,dalCengio2016},
however, suggest to focus on the spatial symmetry of fluctuating physical systems.
This allows to apply the Onsager-Machlup theory, originally restricted to physical
quantities invariant under the time reversal \cite{OM1953,MO1953}, also to currents.

According to the first-order Langevin dynamics, a time autocorrelation function $C_\alpha(t)$
of a fluctuating quantity $\alpha(t)$ is a decaying exponential \cite[Chapter XII]{LLSt}:
\begin{equation}\EQ{exp}
C_\alpha(t) \propto \exp(-\;\mathrm{const}\;|t|)\text{.}
\end{equation}
This analytical result agrees asymptotically with the long-time behavior of the correlations,
found in experiments and computer simulations of classical physical systems \cite[Chapter 2]{Attard}.
Nonetheless, the first-order theory describes {\it inaccurately} correlations at short
times, because the time derivative of \Eq{exp} is discontinuous at $t=0$, whereas
one observes a smooth behavior with $\dot{C}_\alpha(0) = 0$ \cite[Chapter 2]{Attard}.

This failure of the first-order theory to describe correlations at short times can
be attributed to one of its underlying assumptions. Without loss of generality, consider,
for instance, the fluctuations of an equilibrium system described by the ensemble
averages $\avg{\alpha(t)}=0,\avg{\dot{\alpha}(t)} = 0$, cf. \Ref{OM1953}. Suppose,
this system spontaneously fluctuates from an initial {\it complete} \cite[Chapter XII]{LLSt}
equilibrium state $\alpha(0)=0,\,\dot{\alpha}(0)=0$ to another state with $\alpha(t>0)=\alpha_0\ne0,\,\dot{\alpha}(t>0)\ne0$.
The first-order fluctuation theory regards $\dot{\alpha}(t)$ merely as a function\footnote{
One can formally define $\dot{\alpha}[\alpha(t)]$ by using a conditional ensemble
average of $\dot{\alpha}(t)$ at a fixed value $\alpha(t)$, cf. \cite{OnsagerI}.
} $\dot{\alpha}[\alpha(t)]$, i.e. entirely determined by $\alpha(t)$. Physically, this
assumption implies, that the relaxation time from a general transient state $\alpha_0,\,\dot{\alpha}\ne\dot{\alpha}[\alpha_0]$
to the {\it incomplete equilibrium} state $\alpha_0,\,\dot{\alpha}[\alpha_0]$ is
neglected, cf. \cite[Chapter XII]{LLSt}. Therefore, the described {\it quasistationary}
approach \cite[Chapter XII]{LLSt} does not allow to consider the fluctuation dynamics
at arbitrarily short time scales.

The second-order theory of Onsager and Machlup goes beyond the quasistationary approach,
by using two independent variables $\alpha(t)$ and $\dot{\alpha}(t)$ to specify completely
a system state. In the corresponding differential equation, which will be discussed
shortly, a change of state $\alpha(t),\,\dot{\alpha}(t)$ affects only the second
derivative $\ddot{\alpha}(t)$. The fluctuation dynamics then becomes inertial and
is capable of describing transient states, which were ignored in the first-order
theory. One should, therefore, expect that the second-order theory is applicable
to even smaller time scales, than those accessible to the quasistationary approach.

In this paper the original second-order theory of Onsager and Machlup \cite{MO1953}
is first generalized, as suggested in Ref.~\cite{PRE2016II}. In this new form it
becomes, in principle, applicable to both equilibrium and nonequilibrium steady-state
systems. In \Sec{thr} we solve the extended Langevin equation, thus obtained, for
$\alpha(t)$ and $\dot{\alpha}(t)$ and then derive an analytical expression for the
time autocorrelation function $C_\alpha(t)$, which is analogous to \Eq{exp}. Finally,
the usefulness of these theoretical results is demonstrated in an applied study of
shear flow correlations by means of computer simulations in \Sec{sim}.

The time autocorrelation function, which follows from the second-order fluctuation
theory, turns out to describe very accurately the results of our computer simulations
and has the following form:
\begin{equation}\EQ{adv}
C_\alpha(t) \propto \exp\left( -\frac{a |t|}{2} \right) \left[
\cosh\left( \frac{d |t|}{2} \right) + \frac{a}{d} \sinh\left( \frac{d |t|}{2} \right)
\right]\text{,}\\\end{equation}
where $a$ and $d$ are constants to be yet specified in \Sec{thr}. Agreement of \Eq{adv}
with our computational data is observed not only for the long-time behavior, which
remains exponential in character, but also for the short times. In particular, the
time derivative of \Eq{adv} is continuous at $t=0$ with the expected value $\dot{C}_\alpha(0)=0$.
This improvement over the quasistationary approach, as discussed earlier, can be
explained by the finer timescale resolution of the second-order fluctuation theory.

In \Sec{sim} we demonstrate, by using our equilibrium simulations, one practical
application of the analytical form, \Eq{adv}, for the current autocorrelation function.
Namely, we evaluate its time integral in a Green-Kubo formula for the shear viscosity
coefficient \cite[Chapter 7]{Haile}. In principle, this estimation method of the
transport coefficient is more accurate than the usually employed procedure of numerical
integration, as will be discussed.

Our computations show, that the parameters of the normalized current autocorrelation
function are effectively independent of the external shear rate in the linear nonequilibrium
regime. Their values agree with the ones found from our equilibrium simulations.
This is consistent with the fact, that the shear viscosity, which is constant in
the linear nonequilibrium regime, is related to the parameters of the normalized
current autocorrelation function, see \Sec{sim}.

\SEC{thr}{Theory}

The linear Langevin equation of second order in time, proposed by Onsager and Machlup \cite{MO1953}
for a fluctuating quantity $\alpha(t)$, can be expressed in the following general
form:
\begin{equation}\EQ{langevin}
\frac{d^2\alpha(t)}{dt^2} + a \frac{d\alpha(t)}{dt} + b^2 \alpha(t) = \epsilon(t)\text{,}
\end{equation}
where $a > 0$ and $b > 0$ are constants, while $\epsilon(t)$ is a random noise.

The left hand side (LHS) of \Eq{langevin} is analogous to the damped harmonic oscillator.
As it also will become clear from the solution of this equation later in this section,
the parameter $a$ is a friction-like coefficient, which ensures an exponential relaxation
to the macroscopically observable steady state $\avg{\alpha(t)}$ with $\avg{\dot{\alpha}(t)}=0$,
as well as a decay of correlations at long times. The potential-like term, proportional
to $b^2$, determines the resistance of the system to spontaneous fluctuations and
external forces, both due to the right hand side (RHS) of \Eq{langevin}. The constant
$b$, which is analogous to the frequency of the harmonic oscillator, also affects
the autocorrelation function at short times.

In the original theory of Onsager and Machlup for equilibrium systems, the stochastic
part of \Eq{langevin}, i.e. its RHS, which represents irregular spontaneous fluctuating
dynamics due to the ignored degrees of freedom, was assumed Gaussian. Several generalization
of $\epsilon(t)$ were recently announced \cite{KSSH2015II,MorgadoQ2016,PRE2016II}
for the nonequilibrium states, in order to incorporate the action of an external
force. Although in Appendices~\ref{sec:cmf} and \ref{sec:crr} we will treat a more
general case, in this section we follow \Ref{PRE2016II}, by considering a non-Gaussian
random noise of the form:
\begin{equation}\EQ{epsilon}
\epsilon(t) = A dW(t)/dt + B dE(t/\tau)/dt\text{,}
\end{equation}
where $A>0$ and $B\gtreqless0$ are constants, $dW(t)$ and $dE(t/\tau)$ are, respectively,
white noise and exponential noise with a timescale $\tau$, see \Ref{PRE2016II}.

In Equation~(\ref{eq:epsilon}) \cite{PRE2016II}, the constant $A$ is proportional
to the system's temperature, while the ratio $B/\tau$ is the average value of an
external nonequilibrium force. When $B=0$, \Eq{langevin} naturally reduces to the
equilibrium case with the Gaussian random noise, considered by Onsager and Machlup
in \Ref{MO1953}. The noise terms are defined as stochastic differentials of two random
processes: i) the Gaussian process $W(t)$ with a zero mean and a unit variance, and
ii) the Gamma process $E(t/\tau)$ with a timescale $\tau$ and a unit intensity.

In a steady-state the mean values of the time derivatives $\avg{\dot\alpha(t)}$ and
$\avg{\ddot\alpha(t)}$ must vanish by definition. Therefore by taking the appropriate
ensemble average on both sides of \Eq{langevin}, one can read off immediately the macroscopic
behavior of $\alpha(t)$, cf. \Ref{PRE2016II}:
\begin{equation}\EQ{avg}b^2 \avg{\alpha(t)} = B / \tau\text{,}\end{equation}
which implies that for the equilibrium case one has $\avg{\alpha(t)}=0$, while in
the nonequilibrium steady-state the parameter $b^2$ determines the system's response
to an external force $\avg{\alpha(t)} = B / (\tau b^2)$.

A formal solution of \Eq{langevin} for a general stochastic term $\epsilon(t)$ can
be found in Ref.~\cite[Sec.~II.3]{Chandrasekhar}. For that, in principle, one must
consider three cases of a discriminant $d^2 = a^2 - 4 b^2$: i) a periodic solution
$d^2 < 0$, ii) an aperiodic solution $d^2 = 0$, iii) and a nonperiodic solution $d^2 > 0$.
In this paper we consider only the last one\footnote{
A simple prescription, how to obtain the periodic and aperiodic solutions of \Eq{langevin},
can be found in Ref.~\cite[Sec.~II.3]{Chandrasekhar}.
}, because it is applied later in \Sec{sim} to our simulations. Below we cite, in
a more compact form, the nonperiodic solution of \Eq{langevin} from Ref.~\cite[Sec.~II.3]{Chandrasekhar}:
\begin{eqnarray}
\EQ{alpha}
\alpha(t) - \alpha(0) c(t) + \dot{\alpha}(0) \dot{c}(t) / b^2 &=& \int_0^t ds \phi(t-s) \epsilon(s)
\\ \EQ{dalpha}
\dot\alpha(t) - \alpha(0) \dot{c}(t) + \dot\alpha(0) \ddot{c}(t) / b^2 &=& \int_0^t ds \dot{\phi}(t-s) \epsilon(s)
\text{,}
\end{eqnarray}
where
\begin{widetext}
\begin{eqnarray}\EQ{aux}
c(t) &=& \exp \left(-\frac{a t}{2}\right) \left[
\cosh \left(\frac{\sqrt{a^2-4 b^2}}{2} t\right)
+ \frac{a}{\sqrt{a^2-4 b^2}} \sinh \left(\frac{\sqrt{a^2-4 b^2}}{2} t\right)
\right]
\\\EQ{phi}
\phi(t) &=& (a^2-4 b^2)^{-1/2} \left[
  \exp \left( \frac{-a + \sqrt{a^2-4 b^2}}{2} t \right)
  - \exp \left( \frac{-a - \sqrt{a^2-4 b^2}}{2} t \right)
\right]
\text{.}\end{eqnarray}
\end{widetext}

The problem, which remains to deal with in this paper, is to characterize the probability
distribution of the random variables, given by the right hand sides of Eqs.~(\ref{eq:alpha} and \ref{eq:dalpha}).
For the stochastic noise of the form \Eq{epsilon}, the probability density of these
variables apparently can not be expressed in terms of elementary functions. Nonetheless
a method, already shown in \Ref{PRE2016II}, allows to derive an integral representation
of their cumulant-generating functions\footnote{
A cumulant-generating function of a random variable is the Laplace transform of
its probability density.
} and, therefore, to compute analytically their statistical properties.

In Appendix~\ref{sec:cmf} the cumulant-generating function of $\alpha(t)$ is obtained,
as described just above. However, while discussing the time autocorrelation function
$C_\alpha(t)$, we are mainly concerned with the steady-state (SS) solution of \Eq{langevin}:
$$\alpha_\mathrm{SS} = \lim_{t\to\infty} \alpha(t)\text{,}$$
cumulants of which are also calculated in Appendix~\ref{sec:cmf}. Indeed, the time
autocorrelation function can be written as:
\begin{equation}\EQ{corr}
C_\alpha(t) = \avg{\alpha(0) \alpha(t)} - \avg{\alpha^2(t)} = \kappa_2(\alpha_\mathrm{SS}) c_\alpha(t)
\text{,}\end{equation}
where $\kappa_2(\alpha_\mathrm{SS})$ is the second cumulant of $\alpha_\mathrm{SS}$
or, in other words, its variance, and $c_\alpha(t)$ is the normalized time autocorrelation
function, which corresponds to the Pearson correlation coefficient in the statistical
terminology.

By inspecting Eq.~(\ref{eq:alpha}), one can already see that the time-dependent solution
$\alpha(t)$ contains a memory of its initial state $\alpha(0)$ and$\dot\alpha(0)$.
With time this deterministic contribution is vanishing, so that the stochastic part
due to the RHS of the equation, becomes progressively more dominant. How quickly
$\alpha(t)$ is forgetting its initial value $\alpha(0)$ is controlled by the function $c(t)$.
This observation suggests, that $c(t)$ is related to the time autocorrelation function
of $\alpha(t)$. In fact, in Appendix~\ref{sec:crr} we prove, that $c(t)$ is nothing
else but its correlation coefficient:
\begin{equation}\EQ{crr} c_\alpha(t) = c(t)\text{.}\end{equation}
Equation~(\ref{eq:adv}) follows from the above one due to the time-reversal symmetry
of the autocorrelations, with $d = \sqrt{a^2-4 b^2}$.

Note, that the stochastic part of the second-order Langevin equation determines only
the coefficient of proportionality $\kappa_2(\alpha_\mathrm{SS})$ between $C_\alpha(t)$
and $c_\alpha(t)$ in \Eq{corr}. The normalized autocorrelation function, cf. Eqs.~(\ref{eq:aux}
and \ref{eq:crr}), depends only on the constant parameters of the deterministic terms
in \Eq{langevin}, namely $a$ and $b$. This comes, perhaps, without a surprise, because
the correlation is a measure of mutual deterministic dependence between quantities,
$\alpha(t)$ and $\alpha(0)$ in this case. The stochastic term merely erases the connection
between them, so that $\alpha_\mathrm{SS}$ turns into a completely random variable.
Drawn from this, another conclusion is that the analytical expression of the normalized
autocorrelation function is the same for equilibrium and nonequilibrium cases, since
they differ only by the stochastic part of \Eq{langevin}.

\SEC{sim}{Simulations}

In this section we apply the theory, described in \Sec{thr}, to study time autocorrelations
of a shear flow by means of molecular dynamics simulations. Details of our computational
model can be found in Appendix~\ref{sec:cmp}. Here we just mention, that we consider
a thermostatted Weeks-Chandler-Andersen (WCA) fluid \cite{WCA} in three dimensions,
with a constant shear rate $\gamma$ maintained by the Lees-Edwards periodic boundary
conditions \cite[Chapter 6]{EvansMorriss}. All results are reported in reduced units,
cf. Appendix~\ref{sec:cmp}.

The off-diagonal components of the pressure tensor $P_{xy}$, $P_{yz}$ and $P_{zx}$
are the three fluctuating observables, which we measure. They represent the transverse
currents of the linear momentum, each obeying separately \Eq{langevin} by assumption.
In our simulations of nonequilibrium steady-states, the external shear force $\gamma > 0$
causes an average shear flow $\avg{P_{xy}} < 0$, so that
\begin{equation}\EQ{crel} \avg{P_{xy}} = - \eta \gamma \text{,} \end{equation}
where the transport coefficient $\eta$ is the shear viscosity.

\begin{figure}[t!]
\includegraphics[width=1\columnwidth]{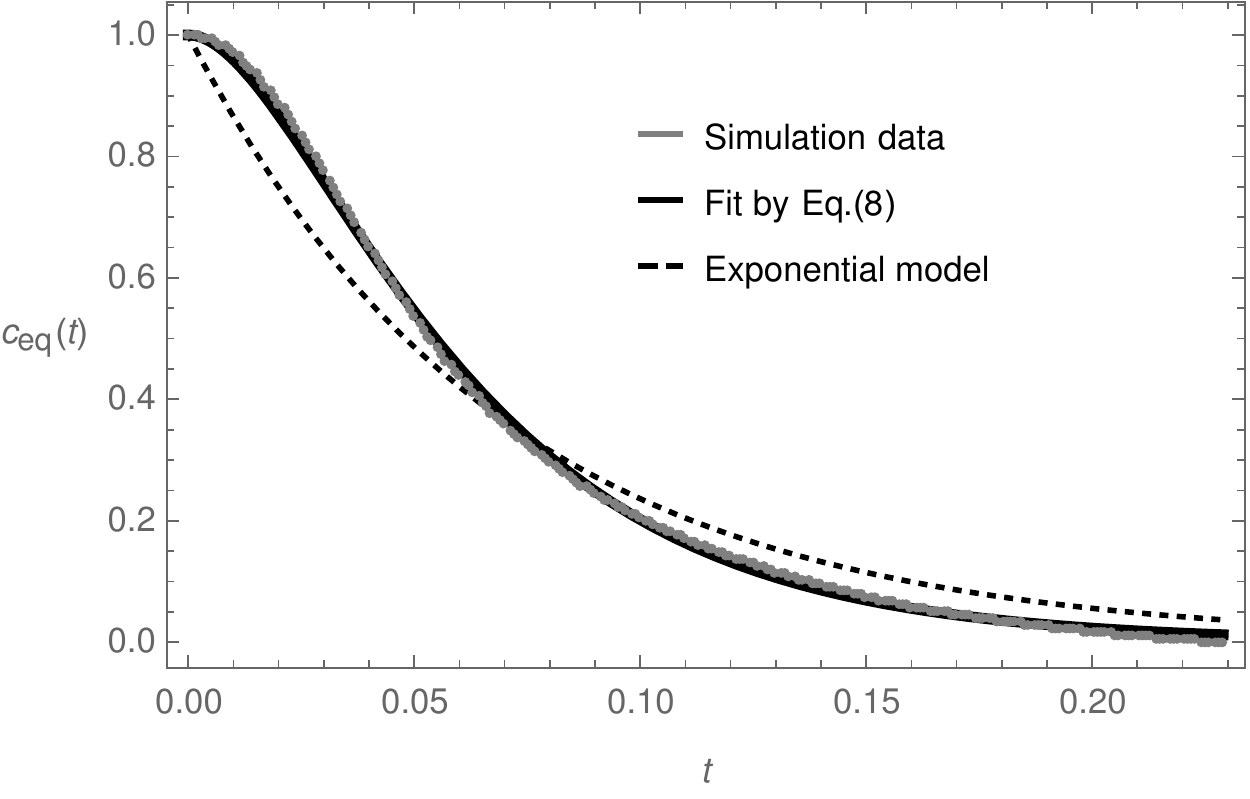}
\FIG{equ}{
Normalized time autocorrelation function $c_\mathrm{eq}(t)$ of shear flow fluctuations
in equilibrium: comparison of the molecular dynamics simulation data with their
fits by analytical models: i) \Eq{aux}, proposed in this paper, ii) the exponential
model of the first-order fluctuation theory \Eq{exp}.
}
\end{figure}

Let us begin, however, with the equilibrium simulations, i.e. $\gamma=0$. In the
absence of an external force, due to symmetry considerations, the statistical properties
of $P_{xy}(t)$, $P_{yz}(t)$ and $P_{zx}(t)$ coincide. In particular, their equilibrium
time autocorrelation functions are equal, respectively, $C_{xy}(t) = C_{yz}(t) = C_{zx}(t) = C_\mathrm{eq}(t)$.
This allows us to exploit more efficiently the statistics of measurements sample
$\{P_{xy}(t_i),P_{yz}(t_i),P_{zx}(t_i)\}_{i=0..n-1}$ ($t_0 = 0$) of a given size
$n$, as follows:
\begin{eqnarray}\EQ{eq}
C_\mathrm{eq}(t_i) &=& [C_{xy}(t_i) + C_{yz}(t_i) + C_{zx}(t_i)] / 3
\\\EQ{all}
C_{xy}(t_i) &\approx& \frac{1}{n-i} \sum_{j=0}^{n-i-1} P_{xy}(t_j) P_{xy}(t_{j+i})
\text{,}\end{eqnarray}
and similarly for $C_{yz}$ and $C_{zx}$. The corresponding equilibrium normalized time autocorrelation
function is
\begin{eqnarray}\EQ{ceq}
c_\mathrm{eq}(t) &=& C_\mathrm{eq}(t) / \kappa_2
\text{,}\end{eqnarray}
where $\kappa_2$ is the variance of the shear flow fluctuations in equilibrium, which
is identical for $P_{xy}$, $P_{yz}$ and $P_{zx}$. For a fixed $n$, the statistical
uncertainty of resulting $C_{xy}(t_i)$ is growing with $i$, since the number of terms
$n-i$ in \Eq{all} is then decreasing. Therefore we restrict the maximum considered
time of correlations $t_\mathrm{max}$ by the widely accepted rule of ``first zero'',
due to \Ref{LagarkovSergeev}.\footnote{
The maximum time $t_i$ does not exceed the first zero of the time autocorrelation
function, i.e. $t_\mathrm{max} < \inf \{t: C_\mathrm{eq}(t_i)=0\}$.
}

\begin{figure}[t!]
\includegraphics[width=1\columnwidth]{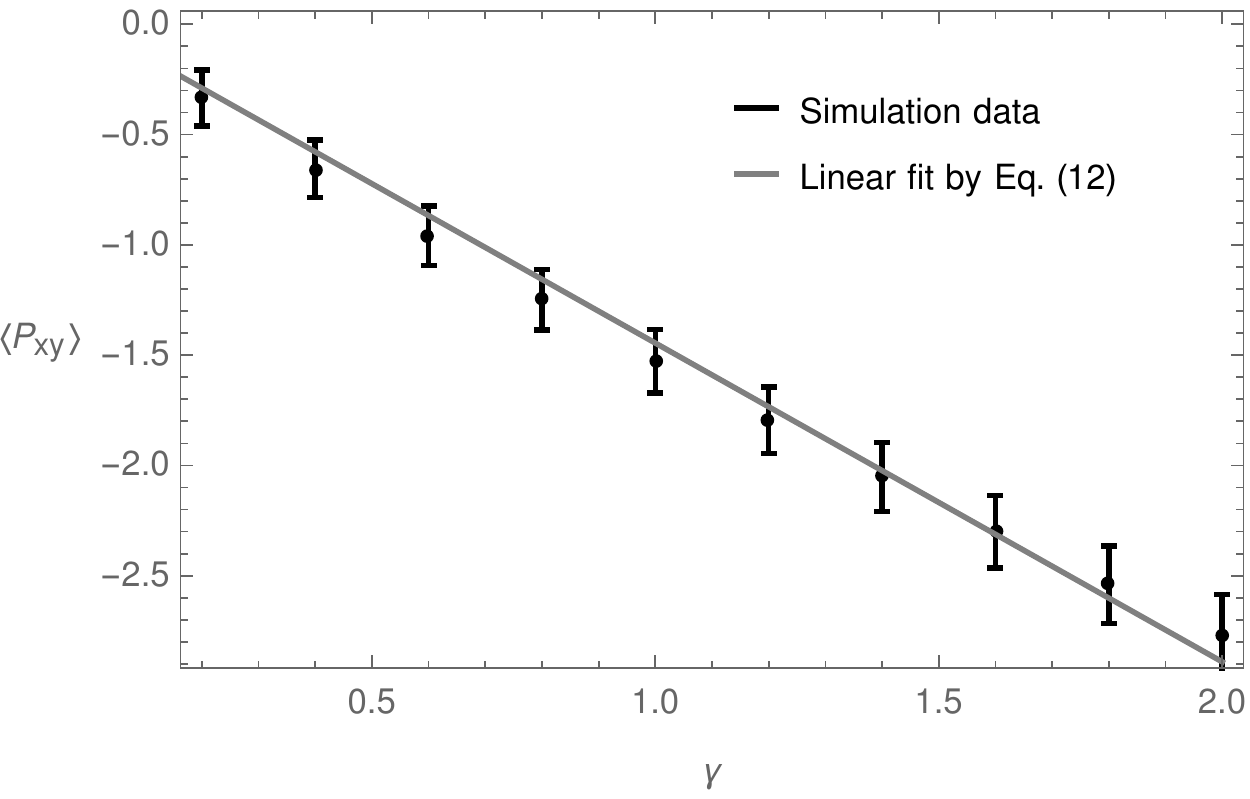}
\FIG{lin}{
Average current $\avg{P_{xy}}$, as a function of the externally applied shear rate
$\gamma$, for our nonequilibrium steady-state simulations. The error bars are given
by three standard deviations. A shear thinning, which corresponds to a decease
of the viscosity with $\gamma$, commonly observed for the WCA fluid, is statistically
insignificant and negligible in the given hydrodynamic regime. Therefore the linear
constitutive relation \Eq{crel}, with a constant viscosity $\eta$, renders a very
good fit to simulation data.
}
\end{figure}

In \Fig{equ} the normalized time autocorrelation function of the shear flow fluctuations,
observed in our equilibrium simulations, is compared with its least-squares fit by
\Eq{aux}. Our analytical model is in excellent agreement with the simulation data,
including the region of short times. The exponential model, \Eq{exp}, which is also
illustrated in \Fig{equ} for comparison, demonstrates the failure of the first order
fluctuation theory for $t\to0$, discussed in \Sec{intro}.

\begin{figure*}[t!]
\includegraphics[width=2\columnwidth]{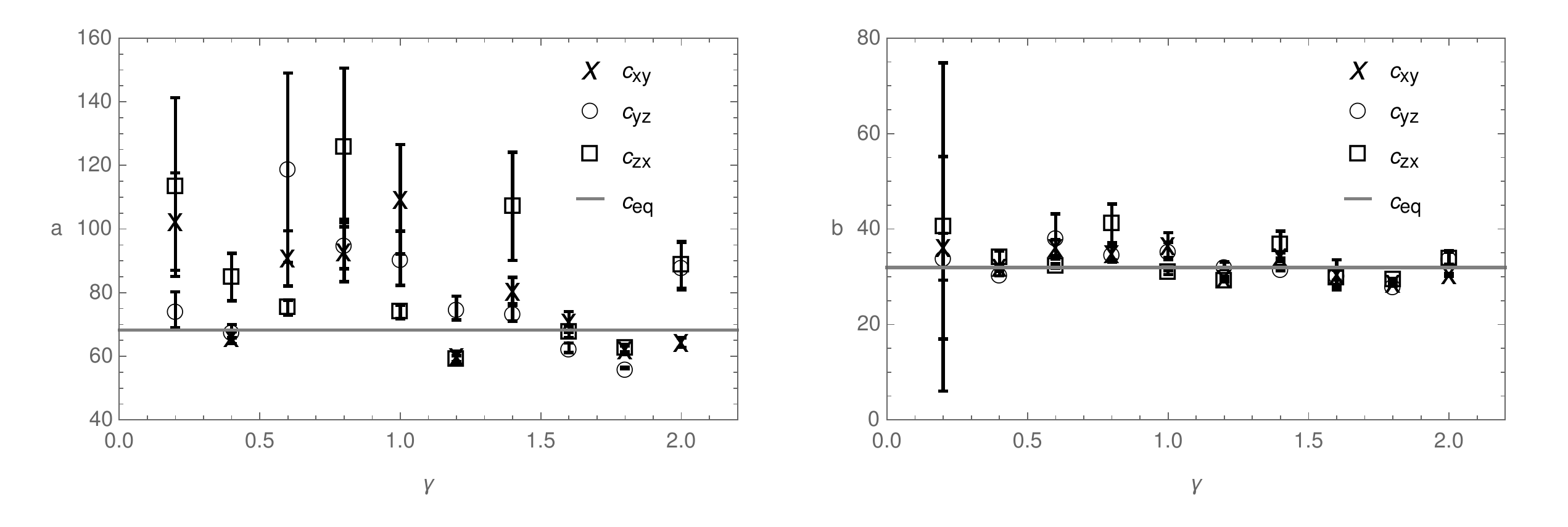}
\FIG{pars}{
Fitting parameters $a$ and $b$ of \Eq{aux} for normalized time autocorrelation
functions at various shear rates in our nonequilibrium steady-state simulations.
Values of the parameters for equilibrium simulations are drawn as solid horizontal
lines for comparison. Error bars are given by three standard deviations.
}
\end{figure*}

Provided that the shear viscosity in \Eq{crel} is independent of $\gamma$, the equilibrium
simulations allow to compute its value, $\eta$, via the Green-Kubo formula \cite[Chapter 6]{EvansMorriss}:
\begin{equation}\EQ{gk} \eta = \frac{V}{k_B T} \int_0^\infty ds C_\mathrm{eq}(t)\text{,} \end{equation}
which together with Eqs.~(\ref{eq:aux}, \ref{eq:corr}, and \ref{eq:crr}) yields
\begin{equation}\EQ{gk1} \eta = \frac{a \kappa_2 V}{k_B T b^2}\text{.}\end{equation}
Equation~(\ref{eq:gk1}) offers an alternative for a numerical approximation of the
integral in \Eq{gk} by a discrete sum:
\begin{equation}\EQ{gk2} \eta = \frac{V \Delta{t}}{k_B T} \sum_{i=0}^{t_\mathrm{max}/\Delta{t}} C_\mathrm{eq}(t_i)\text{,}\end{equation}
where $\Delta{t}$ is the time step between successive measurements.

Indeed, our nonequilibrium simulations show, that the linear regime of constant viscosity
spans a wide range of shear rates $\gamma \in (0, 2.0]$, as demonstrated by \Fig{lin}.
Table~\ref{tab:one} presents estimations of the shear viscosity, computed by various
methods from our simulation data. All results are in very good agreement with each
other. Equation~(\ref{eq:gk2}), though, has a larger statistical uncertainty than
\Eq{gk1}, because the latter interpolates the behavior of the time autocorrelation
function at successive time instants and, therefore, exploits more efficiently the
simulation data. This inference may become even more important, if the time step
$\Delta{t}$ is larger than in our study or the number of measurements $n$ is less.

\begin{table}[b!]
\caption{\label{tab:one} Shear viscosity estimations, computed by various methods from
simulation data.}
\begin{ruledtabular}
\begin{tabular}{ll}
Method & Shear viscosity ($\eta$) \\
\hline
Nonequilibrium simulations, \Eq{crel} & $1.445 \pm 0.020$ \\
Green-Kubo formula, \Eq{gk1} & $1.437 \pm 0.095$ \\
Green-Kubo formula, \Eq{gk2} & $1.44 \pm 0.34$ \\
\end{tabular}
\end{ruledtabular}
\end{table}

Now we turn our attention to the time autocorrelations of shear currents in nonequilibrium
steady states. Since the external force $\gamma$ introduces a preferred spatial direction,
the symmetry argument, which we used for \Eq{eq}, does not apply in this case. Therefore
the nonequilibrium time autocorrelation functions $C_{xy}(t|\gamma)$, $C_{yz}(t|\gamma)$
and $C_{zx}(t|\gamma)$ should be considered separately:
\begin{eqnarray}\EQ{ne}
C_{xy}(t_i|\gamma) &=& \kappa_2(P_{xy}) c_{xy}(t_i|\gamma) \nonumber\\
  &\approx& \frac{1}{n-i}
  \sum_{j=0}^{n-i-1} [
    P_{xy}(t_j) P_{xy}(t_{j+i}) - \avg{P_{xy}}^2
  ]\Bigr|_\gamma
\text{,}\nonumber\\
\end{eqnarray}
and similarly for $C_{yz}(t,\gamma)$ and $C_{zx}(t,\gamma)$, where $c_{xy}(t|\gamma)$
etc. stand for the normalized time autocorrelation functions.

Quite unexpectedly, we found that the parameters of our analytical model for the
normalized autocorrelation functions did not exhibit any notable dependence on the
shear rate. By inspecting \Fig{pars}, which presents results of fitting \Eq{aux}
to the simulation data, one observes no particular difference in the behavior of
the parameters $a$ and $b$ between the three autocorrelation functions of interest.
A remarkable aspect of these plots is that the points with smaller error bars are
all close to the solid horizontal lines, which represent the values of $a$ and $b$
for the equilibrium function $c_\mathrm{eq}(t)$. The large uncertainties of the data,
which occur in the upper part of the graph, suggest that the statistical errors are
biased and cause overestimation of the fitting parameters. This tendency can be explained
by the errors of calculated values $c_{xy}(t_i)$, which are non-identically distributed
and biased, see \cite[Sec. 8.14]{KenneyKeepingII}. Since no clearly visible trend
is observed in \Fig{pars} and the smaller errors of the estimations are close to
the solid lines, we conclude that $a$ and $b$ are practically independent of the
shear rate.

\begin{figure*}[pt!]
\includegraphics[width=2\columnwidth]{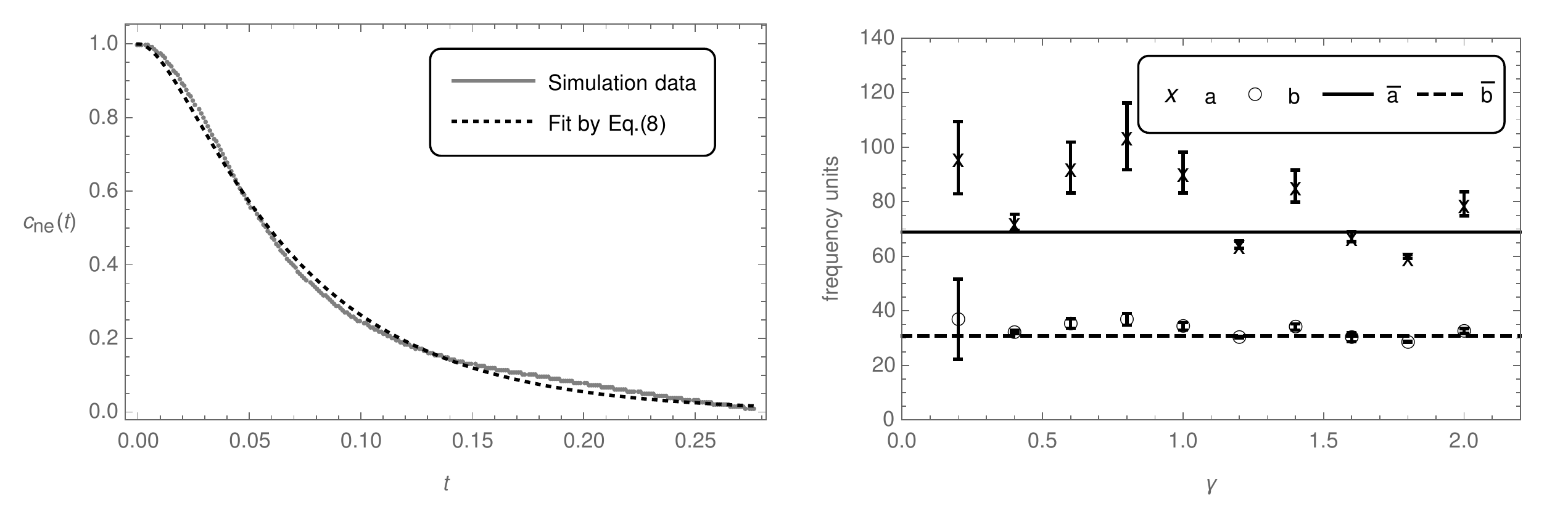}
\FIG{params}{
Fitting \Eq{aux} to the normalized autocorrelation function $c_\mathrm{ne}(t)$,
observed in our nonequilibrium simulations. Left panel: example of $c_\mathrm{ne}(t)$ for
one of our simulations at the shear rate $\gamma=1$. Right panel: the fitting parameters
$a$ and $b$ of \Eq{aux}, computed for our simulations at various shear rates; the
solid horizontal lines $\bar{a}$ and $\bar{b}$ are, respectively, the average values
of $a$ and $b$, weighted by their standard deviations; error bars are given by
three standard deviations.
}
\end{figure*}

The above observations suggest to use the following averaging procedure for the normalized
autocorrelation function of the nonequilibrium steady-state, independently of the shear rate:
\begin{equation}\EQ{ssc}
c_\mathrm{ne}(t_i) = \frac{1}{3} \left[ \frac{C_{xy}(t_i)}{\kappa_2(P_{xy})}
  + \frac{C_{yz}(t_i)}{\kappa_2(P_{yz})} + \frac{C_{zx}(t_i)}{\kappa_2(P_{zx})}
\right]\text{.}
\end{equation}
Although \Eq{ssc} is somewhat similar to \Eq{eq}, the statistical uncertainty of
the latter is much smaller for several reasons. First, the equilibrium autocorrelations
are calculated by taking into account explicitly, that the average current vanishes
$\avg{P_{xy}} = 0$, cf. Eqs.~(\ref{eq:all} and \ref{eq:ne}), while in the nonequilibrium
case one uses the same measurements to evaluate $\avg{P_{xy}}$ and then \Eq{ne}.
Second, measurements of $P_{xy}$, $P_{yz}$ and $P_{zx}$ are merged to compute $\kappa_2$
for the equilibrium fluctuations, whereas $\kappa_2(P_{xy})$ etc. are estimated separately.
Finally, the variance of nonequilibrium fluctuations is greater, cf. \Eq{kappas}
for $B\ne0$. Therefore, under otherwise equal conditions, \Eq{eq} is subject to tighter
statistical constraints, than \Eq{ssc}.

An example of $c_\mathrm{ne}(t)$, constructed for one of our nonequilibrium steady-state
simulations and fitted by \Eq{aux}, can be found in the left panel of \Fig{params}.
A quite good agreement between the computational data and the analytical formula
appears slightly worse, than it was in \Fig{equ} for the equilibrium simulations.
This is due to the larger statistical uncertainties of \Eq{ssc}, as explained above.

The right panel of \Fig{params} presents the fitting parameters of \Eq{aux} for $c_\mathrm{ne}(t)$
calculated for our nonequilibrium simulations at various shear rates. There is a
tendency to overestimation of $a$ and $b$ due to the same error bias, which was noted
earlier in \Fig{pars}. To mitigate this effect, we compute the average values of
the fitting parameters, weighted by their standard deviations:
\begin{equation}\EQ{wavg}\bar{a} = n_\gamma^{-1} \sum_i a(\gamma_i) / \Delta a(\gamma_i)\end{equation}
and similarly for $\bar{b}$, where $n_\gamma^{-1}$ is the number of estimations $a(\gamma_i)$
with the standard deviations $\Delta{a}(\gamma_i)$.

Table~{\ref{tab:two}} compares the final estimations of the parameters in our analytical
model of autocorrelation function \Eq{aux} for equilibrium and nonequilibrium simulations.
The results support our conclusion, that $a$ and $b$ are independent of the shear
rate. The estimations obtained from equilibrium and nonequilibrium data are very
close, although the latter are subject to a larger statistical uncertainty, as was
already discussed.

\begin{table}[b!]
\caption{\label{tab:two} Parameters of the autocorrelation function \Eq{aux}, determined
from our equilibrium and nonequilibrium simulations. The nonequilibrium estimations
are evaluated by the weighted averages, see \Eq{wavg}}
\begin{ruledtabular}
\begin{tabular}{lll}
& $a$ & $b$ \\
\hline
Equilibrium simulations, $c_\mathrm{eq}(t)$ & $68.28 \pm 0.76$ & $31.93 \pm 0.17$ \\
Nonequilibrium simulations, $c_\mathrm{ne}(t)$ & $69 \pm 13$ & $30.8 \pm 2.7$ \\
\end{tabular}
\end{ruledtabular}
\end{table}

Although the independence of the autocorrelation parameters $a$ and $b$ from the
external force is rather unexpected, it may be explained by the constant viscosity
coefficient. Indeed, the shear viscosity is related to $a$ and $b$ by several relations,
cf. Eqs.~(\ref{eq:avg},\ref{eq:crel}, and \ref{eq:gk1}). If the transport coefficient
is constant in the constitutive relation \Eq{crel}, either there must be some peculiar
relations between the constants of \Eq{langevin}, e.g. $b$ and $B/\tau$ due to Eqs.~(\ref{eq:avg} and \ref{eq:crel}),
either some of these constants must remain unaffected by the shear rate. The latter
case, which was observed in our simulations, also appears the more likely of the
two.

\SEC{fin}{Conclusion}

In \Sec{thr} we applied a recent extension of the Langevin equation \cite{PRE2016II},
to generalize the Onsager-Machlup fluctuation theory of the second order in time
for equilibrium and nonequilibrium steady states. A solution technique for this class
of stochastic dynamical problems was demonstrated in Appendix~\ref{sec:cmf}.

The analytical expression of the time autocorrelation function, \Eq{adv}, derived
in Appendix~\ref{sec:crr} for the second-order fluctuation theory, correctly describes
not only the exponential decay of the correlations, but also their smooth behavior
at short times. Our computational study of the hydrodynamic shear flow confirmed
that the generalized Onsager-Machlup theory is applicable to the current fluctuations.
In particular, plugged into the Green-Kubo formula, \Eq{adv} renders excellent results
for the transport coefficient.

In our simulations, we found that the normalized correlation function of shear flow,
specified by the deterministic part of the fluctuation dynamics, is practically independent
of the external shear rate in the linear nonequilibrium regime. This can be attributed
to the connection between the constant transport coefficient, shear viscosity in
our case, and the parameters of \Eq{adv}.

Finally, we would like to remark, that a second-order time derivative introduces
into the Langevin dynamics a dependence on the history of the fluctuating state variable.
This provides a link with the formalism of the memory function \cite[Chapter 4]{AdamoBelousovRondoni,BoonYip},
which is widely used, e.g. in the hydrodynamics. Therefore \Eq{adv} might find further
applications in this context.

\begin{acknowledgments}
One of the authors, Dr. Roman Belousov, is obliged to Dr. Alexei Bazavov for stimulating
discussions on operator splitting techniques for molecular dynamics simulations.
\end{acknowledgments}

\appendix
\SEC{cmf}{Statistics of the second-order fluctuation dynamics}

As described in \Sec{thr}, in order to complete the solution of \Eq{langevin},
which was stated formally by Eqs.~(\ref{eq:alpha} and \ref{eq:dalpha}), one needs
to specify the probability distribution of the following random processes $r(t)$
and $\dot{r}(t)$:
\begin{equation}\EQ{rdr}
r(t) = \int_0^t ds \phi(t-s) \epsilon(s); \quad \dot{r}(t) = \int_0^t ds \dot{\phi}(t-s) \epsilon(s)\text{,}
\end{equation}
where $\epsilon(t)$ is given by \Eq{epsilon}.

In this appendix we will deal with this problem in a slightly more general way than
in \Sec{thr}, so that our result will be valid not only for exponential noise, proposed
in \cite{PRE2016II}. In particular, our solution also applies to some variants of
shot noise \cite{KSSH2015II,MorgadoQ2016,PRE2016II}. More specifically, we will treat
a general stochastic process with stationary independent increments
\begin{equation}\EQ{dR} R(t) = \int_0^t ds \epsilon(s)\text{,}\end{equation}
such that its cumulant-generating function can be represented as:
\begin{equation}\EQ{rep}
\mathcal{C}_R(\tilde{R},t) = t f(\tilde{R})\text{.}
\end{equation}
Here $\tilde{R}$ is the reciprocal dual of $R$, while $f(\cdot)$ is a cumulant-generating
function of some random variable, which we call further an {\it elementary} cumulant-generating
function of $R(t)$. The representation \Eq{rep} applies to the Wiener process, the Gamma process, the
Poisson process or a sum composed of these,\footnote{
The representation \Eq{rep} of the stationary stochastic processes, mentioned in this paper,
is related to their property of {\it infinite divisibility} \cite{Steutel}.
} cf. \cite[Table I]{PRE2016II}.

From now onwards, reciprocal duals of random variables are denoted by tilde, like
we already used above $\tilde{R}$ for the dual of $R$. To proceed, we will need the
following lemma.

\begin{lemma}\label{lemma}
  Let $R$ and $\tilde{R}$ be, respectively, a real random variable and its reciprocal
  dual with a cumulant-generating function $\mathcal{C}_R(\tilde{R})$. Then a joint
  cumulant-generating function of random variables $\rho_1 = c_1 R$ and $\rho_2 = c_2 R$,
  where $c_1$ and $c_2$ are real constants, is given by
  \begin{equation}\EQ{lemma}
    \mathcal{C}_\rho(\tilde\rho_1, \tilde\rho_2) = \mathcal{C}_R(c_1 \tilde\rho_1 + c_2 \tilde\rho_2)
  \text{,}\end{equation}
  where $\tilde\rho_1$ and $\tilde\rho_2$ are reciprocal duals of $\rho_1$ and $\rho_2$,
  respectively.
\end{lemma}
\begin{proof}[Proof]
  Let us introduce an auxiliary random variable $\prm{R} = R$. Then a joint probability
  density of $R$ and $\prm{R}$ is $p(R, \prm{R}) = p(R) \delta(R-\prm{R})$, where
  $p(R)$ is the probability density of $R$ and $\delta(\cdot)$ is the Dirac delta-function.
  The joint probability measure of $\rho_1$ and $\rho_2$ is
  $$p_\rho(\rho_1,\rho_2) d\rho_1 d\rho_2 = p(\rho_1/c_1, \rho_2/c_2) \frac{d\rho_1}{c_1} \frac{d\rho_2}{c_2}\text{,}$$
  whence we have:
  \begin{eqnarray}\EQ{proof}
    \mathcal{C}_\rho(\tilde\rho_1, \tilde\rho_2)
      &=& \ln \iint d\rho_1 d\rho_2 \exp(\rho_1 \tilde\rho_1 + \rho_2 \tilde\rho_2) p_\rho(\rho_1, \rho_2)
    \nonumber\\
      = \ln \iint &dR& \prm{dR} \exp(c_1 R \tilde\rho_1 + c_2 \prm{R} \tilde\rho_2) p(R) \delta(R - \prm{R})
    \nonumber\\
      = \ln \iint &dR& \exp[R (c_1 \tilde\rho_1 + c_2 \tilde\rho_2)] p(R)
    \nonumber\\
      &=& \mathcal{C}_R(c_1 \tilde\rho_1 + c_2 \tilde\rho_2)
  \text{,}\end{eqnarray}
  which proves the lemma.
\end{proof}

Now we express $r(t)$ and $\dot{r}(t)$ from \Eq{rdr}, as limits of the following
discrete sums, $s$ and $\dot{s}_n$, respectively, by partitioning the domain $[0,t]$
into $n$ subintervals of length $\Delta{t}$ so that $n \Delta{t} = t$:
\begin{eqnarray}\EQ{sum}
S(n) &=& \sum_{j=0}^{n-1} \phi(t-j\Delta{t})
  \int_{j\Delta{t}}^{(j+1)\Delta{t}} ds \epsilon(s) \nonumber\\
  &=& \sum_{j=0}^{n-1} \phi(t-j\Delta{t}) R(\Delta{t})
  = \sum_{j=0}^{n-1} R_j \underset{n\to\infty}{\to} r(t)
\text{,}\end{eqnarray}
where we used \Eq{dR}, and similarly,
\begin{equation}\EQ{dsum}
\dot{S}(n) = \sum_{j=0}^{n-1} \dot{\phi}(t-j\Delta{t}) R(\Delta{t}) = \sum_{j=0}^{n-1} \dot{R}_j \underset{n\to\infty}{\to} \dot{r}(t)\text{.}
\end{equation}

By assumption, $R(t)$ has independent stationary increments. Therefore the terms
$R_j$ in \Eq{sum} are mutually independent. The same argument also applies to $\dot{R}_j$
in \Eq{dsum}. Consequently, the joint cumulant-generating function $\mathcal{C}_{S\dot{S}}(\cdot, \cdot, n)$
for $S$ and $\dot{S}$ equals the sum of the joint cumulant-generating functions of $R_j$ and $\dot{R}_j$,
$\mathcal{C}_j(\cdot,\cdot)$. By applying Lemma~\ref{lemma} to the latter
and using the representation \Eq{rep}, this yields:
\begin{eqnarray}\EQ{cmft}
\mathcal{C}_{S\dot{S}}(\tilde{S}, \tilde{\dot{S}}, n) &=& \sum_{j=0}^{n-1} \mathcal{C}_j(\tilde{S},\tilde{\dot{S}})
  \nonumber\\
  &=& \sum_{j=0}^{n-1} \mathcal{C}_R[\phi(t-j\Delta{t}) \tilde{S} + \dot\phi(t-j\Delta{t}) \tilde{\dot{S}}, \Delta{t}]
  \nonumber\\
  &\underset{n\to\infty}{\to}&
  \int_0^t ds f[\phi(t-s) \tilde{r} + \dot\phi(t-s) \tilde{\dot{r}}]
  \nonumber\\
  &=& \mathcal{C}_{r\dot{r}}(\tilde{r},\tilde{\dot{r}},t)\text{,}
\end{eqnarray}
which is the joint cumulant-generating function of $r(t)$ and $\dot{r}(t)$. By repeating
the above derivation for the marginal distributions of $r(t)$ and $\dot{r}(t)$, one can
obtain separately their cumulant-generating functions:
\begin{eqnarray}\EQ{cmfr}
\mathcal{C}_r(\tilde{r},t) &=& \int_0^t ds f[\phi(t-s) \tilde{r}]
\\\EQ{cmfdr}
\mathcal{C}_{\dot{r}}(\tilde{\dot{r}},t) &=& \int_0^t ds f[\dot\phi(t-s) \tilde{\dot{r}}]
\text{.}\end{eqnarray}

The integral representations Eqs.~(\ref{eq:cmft}-\ref{eq:cmfdr}) allow to compute
analytically the cumulants of $r$ and $\dot{r}$, once the elementary cumulant-generating
function $f$ is specified. Since in \Sec{thr} and Appendix~\ref{sec:crr} we are mainly
interested in the steady-state solution of \Eq{langevin}, $\alpha_\mathrm{SS}$ and
$\dot\alpha_\mathrm{SS}$, below we calculate their cumulants $\kappa_i,\,i=1..3$:
\begin{eqnarray}\EQ{cum}
\kappa_1(\alpha_\mathrm{SS}) &=& \lim_{t\to\infty} \frac{\partial \mathcal{C}_r}{\partial r}(0,t) = \frac{\prm{f}(0)}{b^2}
\nonumber\\
\kappa_2(\alpha_\mathrm{SS}) &=& \lim_{t\to\infty} \frac{\partial^2 \mathcal{C}_r}{\partial r^2}(0,t) = \frac{\prmprm{f}(0)}{2 a b^2}
\nonumber\\
\kappa_3(\alpha_\mathrm{SS}) &=& \lim_{t\to\infty} \frac{\partial^3 \mathcal{C}_r}{\partial r^3}(0,t) = \frac{2 f^{(3)}(0)}{3 b^2 (2 a^2+b^2)}
\nonumber\\
\kappa_1(\dot\alpha_\mathrm{SS}) &=& \lim_{t\to\infty} \frac{\partial \mathcal{C}_{\dot{r}}}{\partial r}(0,t) = 0
\nonumber\\
\kappa_2(\dot\alpha_\mathrm{SS}) &=& \lim_{t\to\infty} \frac{\partial^2 \mathcal{C}_{\dot{r}}}{\partial r^2}(0,t) = \frac{\prmprm{f}(0)}{2 a}
\nonumber\\
\kappa_3(\dot\alpha_\mathrm{SS}) &=& \lim_{t\to\infty} \frac{\partial^3 \mathcal{C}_{\dot{r}}}{\partial r^3}(0,t) = \frac{2 a f^{(3)}(0)}{3 (2 a^2+b^2)}
\text{,}\end{eqnarray}
where $\prm{f}$, $\prmprm{f}$, $f^{(3)}$ denote the derivatives of the elementary
cumulant-generating function. For the stochastic noise, defined by \Eq{epsilon},
we have
$$f(\tilde{R}) = A^2 \tilde{R} / 2 - \ln( 1 - B \tilde{R} )/\tau \text{,}$$
which, due to \Eq{cum}, renders the following cumulants of $\alpha_\mathrm{SS}$
\begin{eqnarray}\EQ{kappas}
\nonumber\\
&\kappa_3&(\alpha_\mathrm{SS}) = \frac{4 B^3}{3 b^2 \tau (2 a^2+b^2)}\text{.}\\\nonumber
\end{eqnarray}

\SEC{crr}{Time autocorrelation function of the second-order fluctuation dynamics}

In this section we prove \Eq{crr} by using the results of Appendix~\ref{sec:cmf}.
For convenience we adopt a subscript notation $\alpha_t = \alpha(t)$. Consider the
joint steady-state probability function of $\alpha_0$ and $\alpha_t$, which can
be written as:
\begin{equation}\EQ{prb}
p_{0,t}(\alpha_0, \alpha_t) = \int d\dot{\alpha}_0 p_t(\alpha_t|\alpha_0,\dot{\alpha}_0) p(\alpha_0,\dot{\alpha}_0)\text{,}
\end{equation}
where $p(\cdot,\cdot)$ is the joint steady-state probability density of $\alpha_\mathrm{SS}$
and $\dot\alpha_\mathrm{SS}$, while $p_t(\alpha_t|\alpha_0,\dot\alpha_0)$ is the
transition probability to the state $\alpha_t$, conditional on some initial state
$\alpha_0,\,\dot\alpha_0$. But, due to \Eq{alpha}, we have
\begin{equation}\EQ{ptr}
p_t(\alpha_t|\alpha_0,\dot\alpha_0) = p_r[\alpha_t - \alpha_0 c(t) + \dot\alpha_0 \dot{c}(t) / b^2]\text{,}
\end{equation}
where $p_r(\cdot)$ is the probability density of $r(t)$ from \Eq{rdr}, which corresponds
to the cumulant-generating function $\mathcal{C}_r(\cdot,t)$ from \Eq{cmfr}.

By using Eqs.~(\ref{eq:prb} and \ref{eq:ptr}), for the probability distribution $p_{0,t}(\cdot,\cdot)$
we obtain the following joint cumulant-generating function:
\begin{widetext}
\begin{eqnarray}\EQ{joint}
\mathcal{C}_{0,t}(\alpha_0,\alpha_t) &=& \ln \iint d\alpha_0 d\alpha_t \exp(\alpha_0 \tilde\alpha_0 + \alpha_t \tilde\alpha_t) p_{0,t}(\alpha_0, \alpha_t)
\nonumber\\
&=& \ln \iiint d\dot\alpha_0 d\alpha_0 d\alpha_t \exp(\alpha_0 \tilde\alpha_0 + \alpha_t \tilde\alpha_t)
    p(\alpha_0, \dot\alpha_0) p_r[\alpha_t - \alpha_0 c(t) + \dot\alpha_0 \dot{c}(t) / b^2]
\nonumber\\
&=& \mathcal{C}_r(\tilde\alpha_t, t) + \ln \iint d\dot\alpha_0 d\alpha_0
    \exp\{\alpha_0 [\tilde\alpha_0 + \tilde\alpha_t c(t)] - \dot\alpha_0 \tilde\alpha_t \dot{c}(t) / b^2\}
    p(\alpha_0, \dot\alpha_0)
\nonumber\\
&=& \mathcal{C}_r(\tilde\alpha_t, t) + \lim_{\prm{t}\to\infty}
    \mathcal{C}_{r\dot{r}}[\tilde\alpha_0 + \tilde\alpha_t c(t), -\tilde\alpha_t \dot{c}(t) / b^2, \prm{t}]
\text{,}
\end{eqnarray}
where $\mathcal{C}_{r\dot{r}}(\cdot,\cdot,t)$ is given by \Eq{cmft}. The normalized autocorrelation
function $c_\alpha(t)$ follows from Eqs.~(\ref{eq:cmft} and \ref{eq:joint}):
\begin{eqnarray}
c_\alpha(t) &=& \frac{1}{\kappa_2(\alpha_\mathrm{SS})} \frac{\partial^2 \mathcal{C}_{0,t}}{\partial\alpha_t \partial\alpha_0}(0,0)
  = \frac{1}{\kappa_2(\alpha_\mathrm{SS})} \lim_{\prm{t}\to\infty}
    \frac{\partial^2 \mathcal{C}_{r\dot{r}}[\tilde\alpha_0 + \tilde\alpha_t c(t), -\tilde\alpha_t \dot{c}(t) / b^2, \prm{t}]}{
      \partial\alpha_t \partial\alpha_0}\Biggr|_{\substack{\alpha_0=0\\\alpha_t=0}}
  \nonumber\\
  &=& \frac{c(t)\prmprm{f}(0)}{\kappa_2(\alpha_\mathrm{SS})} \lim_{\prm{t}\to\infty} \left[ \int_0^{\prm{t}} ds \phi(\prm{t}-s)^2\right]
    + \frac{\dot{c}(t)\prmprm{f}(0)}{b^2 \kappa_2(\alpha_\mathrm{SS})} \lim_{\prm{t}\to\infty} \left[\int_0^{\prm{t}} ds \phi(\prm{t}-s)\dot\phi(\prm{t}-s)\right]
    = c(t){,}
\end{eqnarray}
\end{widetext}
which concludes our proof of \Eq{crr}.

\SEC{cmp}{Details of the computational model}

In our molecular dynamics simulations, we integrated numerically equations of motion
in $D=3$ dimensions for a system of $N=10000$ particles, which were interacting through
the Weeks-Chandler-Anderson potential \cite{WCA}:
$$
U_\mathrm{WCA}(r) = \begin{cases}
4 \epsilon \left[
  (\frac{\sigma}{r})^{12} - (\frac{\sigma}{r})^{6}
\right]\text{, if } r < 2^{1/6}\sigma \\
0\text{, if } r \ge 2^{1/6}\sigma
\end{cases}\text{,}
$$
where $r$ is the interparticle distance, $\epsilon$ and $\sigma$ are constants of
the potential energy and its range,respectively.

The system was subject to moving periodic boundary conditions, which impose a constant
shear rate $\gamma$ \cite[Chapter 6]{EvansMorriss}, and coupled to the Nos\'{e}-Hoover
(NH) thermostat \cite[Chapter 6]{FrenkelSmit} of the relaxation time constant $\theta$.
The resulting thermostatted SLLOD \cite[Chapter 6]{EvansMorriss} equations of motion
read:
\begin{eqnarray}\EQ{SLLOD}
\dot{\pnt{q}}_i &=& \pnt{p}_i / m + \gamma q_{i y} \pnt{X} \nonumber\\
\dot{\pnt{p}}_i &=& \pnt{F}_i(\pnt{q}_i) - \gamma p_{i y} \pnt{X} - \zeta \pnt{p}_i \nonumber\\
\dot{\zeta} &=& \theta^{-2} \sum_{i=1}^{N} (\frac{\pnt{p}_i^2}{m D N k_B T} - 1)\text{.}
\end{eqnarray}
Here $\pnt{X}$ is a unit vector along the Cartesian coordinate axis $X$ and $\pnt{q}_i$
is the position of $i$-th particle ($q_{i y}$ being its $Y$-coordinate); all particles
have equal mass $m$; $\pnt{p}_i$ is the {\it peculiar} linear momentum of the $i$-th
particle ($p_{i y}$ being its $Y$-component); $\pnt{F}_i$ is the force on the $i$-th
particle due to the interactions with all the other particles; $\zeta$ is the coupling
to the NH reservoir at temperature $T$, while $k_B$ is the Boltzmann constant. The
pressure tensor components are calculated by the following formula:
$$P_{xy} = V^{-1}\sum_{i=1}^{N} (p_{ix} p_{iy}/m + F_{ix} q_{iy})\text{,}$$
and similarly for $P_{yz}$ and $P_{zy}$.

The results of \Sec{sim} are reported in simulation units, reduced by the energy
constant $\epsilon$, the length constant $\sigma$, the mass constant $m$ and the
time constant $\theta$. Invariant parameters of our computational experiments were
the temperature of the NH thermostat $k_B T = 1$ and the number density $0.8$.

The numerical integration of \Eq{SLLOD} was performed with a time step $\Delta{t}=10^{-3}$
by using an optimized version of the symplectic operator-splitting method, proposed
in \Ref{MartynaTuckermanII}. In more detail, we consider an evolution operator, acting
on the extended phase space of points $\pnt{\Gamma} = (\pnt{q}_{1..N},\pnt{p}_{1..N},\zeta)$, so that
$$\pnt{\Gamma}(t) = \exp(\mathrm{i} \mathcal{L} t) \pnt{\Gamma}(0)\text{,}$$
where $\mathrm{i} \mathcal{L} = \dot{\pnt\Gamma}\cdot\nabla_{\pnt\Gamma}$ is the
Liouville operator for \Eq{SLLOD}. The Liouvillian can be split by using the following
operatorial sum:
\begin{eqnarray}
\mathcal{L} &=& \mathcal{L}_{p\gamma} + \mathcal{L}_{p} + \mathcal{L}_{p\zeta} + \mathcal{L}_{q\gamma} + \mathcal{L}_{q\zeta}
\nonumber\\
\mathrm{i} \mathcal{L}_{q\zeta} &=& \frac{\pnt{p}}{m} \cdot \partial_{\pnt{q}} + \dot{\zeta} \partial_\zeta
\nonumber\\
\mathrm{i} \mathcal{L}_{q\gamma} &=& \gamma \pnt{q}_y \cdot \frac{\partial}{\partial\pnt{q}_x}
\nonumber\\
\mathrm{i} \mathcal{L}_{p\zeta} &=& -\zeta \pnt{p} \cdot \partial_{\pnt{p}}
\nonumber\\
\mathrm{i} \mathcal{L}_p &=& \pnt{F}(\pnt{q}) \cdot \partial_{\pnt{p}}
\nonumber\\
\mathrm{i} \mathcal{L}_{p\gamma} &=& - \gamma \pnt{p}_y \cdot \frac{\partial}{\partial\pnt{p}_x}
\text{,}
\end{eqnarray}
where the operators $\partial_{\pnt{q}}$, $\partial_{\pnt{p}}$ etc. act on the respective
subspaces of positions $\pnt{q}$, momenta $\pnt{p}$ etc. in $\pnt{\Gamma}$. The evolution
operator is then approximated by
\begin{widetext}
\begin{eqnarray}\EQ{alg}
\exp[\mathrm{i} \mathcal{L} t + \mathrm{O}(t^2)] &=& \prod_{j=1}^{t/\Delta{t}}
  \exp\left(\frac{\mathrm{i} \mathcal{L}_{p\gamma}\Delta{t}}{2}\right)
  \exp\left(\frac{\mathrm{i} \mathcal{L}_{p} \Delta{t}}{2}\right)
  \exp\left(\frac{\mathrm{i} \mathcal{L}_{p\zeta} \Delta{t}}{2}\right)
  \exp\left(\frac{\mathrm{i} \mathcal{L}_{q\gamma} \Delta{t}}{2}\right)
  \times \nonumber\\
  &\times&\exp(\mathrm{i} \mathcal{L}_{q\zeta} \Delta{t})
  \exp\left(\frac{\mathrm{i} \mathcal{L}_{q\gamma} \Delta{t}}{2}\right)
  \exp\left(\frac{\mathrm{i} \mathcal{L}_{p\zeta} \Delta{t}}{2}\right)
  \exp\left(\frac{\mathrm{i} \mathcal{L}_{p\gamma} \Delta{t}}{2}\right)
  \exp\left(\frac{\mathrm{i} \mathcal{L}_{p} \Delta{t}}{2}\right)
\text{.}\end{eqnarray}
\end{widetext}
The decomposition of the evolution operator in \Eq{alg} determines the sequence of
steps in our symplectic integrator, as described in Refs~\cite[Appendix E]{MartynaTuckermanI,MartynaTuckermanII,FrenkelSmit}.

Before collecting simulation data, the initially generated phase space configurations
were evolved for a time interval of $10^5$ reduced units. Then, in order to calculate
the autocorrelation functions of interest, we measured the pressure tensor components
at $10^5$ consecutive integration steps.

\bibliographystyle{apsrev4-1}
\bibliography{References}

\begin{thebibliography}{24}%
\makeatletter
\providecommand \@ifxundefined [1]{%
 \@ifx{#1\undefined}
}%
\providecommand \@ifnum [1]{%
 \ifnum #1\expandafter \@firstoftwo
 \else \expandafter \@secondoftwo
 \fi
}%
\providecommand \@ifx [1]{%
 \ifx #1\expandafter \@firstoftwo
 \else \expandafter \@secondoftwo
 \fi
}%
\providecommand \natexlab [1]{#1}%
\providecommand \enquote  [1]{``#1''}%
\providecommand \bibnamefont  [1]{#1}%
\providecommand \bibfnamefont [1]{#1}%
\providecommand \citenamefont [1]{#1}%
\providecommand \href@noop [0]{\@secondoftwo}%
\providecommand \href [0]{\begingroup \@sanitize@url \@href}%
\providecommand \@href[1]{\@@startlink{#1}\@@href}%
\providecommand \@@href[1]{\endgroup#1\@@endlink}%
\providecommand \@sanitize@url [0]{\catcode `\\12\catcode `\$12\catcode
  `\&12\catcode `\#12\catcode `\^12\catcode `\_12\catcode `\%12\relax}%
\providecommand \@@startlink[1]{}%
\providecommand \@@endlink[0]{}%
\providecommand \url  [0]{\begingroup\@sanitize@url \@url }%
\providecommand \@url [1]{\endgroup\@href {#1}{\urlprefix }}%
\providecommand \urlprefix  [0]{URL }%
\providecommand \Eprint [0]{\href }%
\providecommand \doibase [0]{http://dx.doi.org/}%
\providecommand \selectlanguage [0]{\@gobble}%
\providecommand \bibinfo  [0]{\@secondoftwo}%
\providecommand \bibfield  [0]{\@secondoftwo}%
\providecommand \translation [1]{[#1]}%
\providecommand \BibitemOpen [0]{}%
\providecommand \bibitemStop [0]{}%
\providecommand \bibitemNoStop [0]{.\EOS\space}%
\providecommand \EOS [0]{\spacefactor3000\relax}%
\providecommand \BibitemShut  [1]{\csname bibitem#1\endcsname}%
\let\auto@bib@innerbib\@empty
\bibitem [{\citenamefont {Onsager}\ and\ \citenamefont
  {Machlup}(1953)}]{OM1953}%
  \BibitemOpen
  \bibfield  {author} {\bibinfo {author} {\bibfnamefont {L.}~\bibnamefont
  {Onsager}}\ and\ \bibinfo {author} {\bibfnamefont {S.}~\bibnamefont
  {Machlup}},\ }\href {\doibase 10.1103/PhysRev.91.1505} {\bibfield  {journal}
  {\bibinfo  {journal} {Phys. Rev.}\ }\textbf {\bibinfo {volume} {91}},\
  \bibinfo {pages} {1505} (\bibinfo {year} {1953})}\BibitemShut {NoStop}%
\bibitem [{\citenamefont {Machlup}\ and\ \citenamefont
  {Onsager}(1953)}]{MO1953}%
  \BibitemOpen
  \bibfield  {author} {\bibinfo {author} {\bibfnamefont {S.}~\bibnamefont
  {Machlup}}\ and\ \bibinfo {author} {\bibfnamefont {L.}~\bibnamefont
  {Onsager}},\ }\href {\doibase 10.1103/PhysRev.91.1512} {\bibfield  {journal}
  {\bibinfo  {journal} {Phys. Rev.}\ }\textbf {\bibinfo {volume} {91}},\
  \bibinfo {pages} {1512} (\bibinfo {year} {1953})}\BibitemShut {NoStop}%
\bibitem [{\citenamefont {Landau}\ \emph {et~al.}(1980)\citenamefont {Landau},
  \citenamefont {Lifshitz},\ and\ \citenamefont {Pitaevskii}}]{LLSt}%
  \BibitemOpen
  \bibfield  {author} {\bibinfo {author} {\bibfnamefont {L.}~\bibnamefont
  {Landau}}, \bibinfo {author} {\bibfnamefont {E.}~\bibnamefont {Lifshitz}}, \
  and\ \bibinfo {author} {\bibfnamefont {L.}~\bibnamefont {Pitaevskii}},\
  }\href@noop {} {\emph {\bibinfo {title} {{S}tatistical {P}hysics}}},\
  \bibinfo {edition} {3rd}\ ed.\ (\bibinfo  {publisher} {Pergamon Press},\
  \bibinfo {address} {Oxford, New York},\ \bibinfo {year} {1980})\BibitemShut
  {NoStop}%
\bibitem [{\citenamefont {Attard}(2012)}]{Attard}%
  \BibitemOpen
  \bibfield  {author} {\bibinfo {author} {\bibfnamefont {P.}~\bibnamefont
  {Attard}},\ }\href@noop {} {\emph {\bibinfo {title} {{N}on-{E}quilibrium
  {T}hermodynamics and {S}tatistical {M}echanics}}}\ (\bibinfo  {publisher}
  {Oxford University Press},\ \bibinfo {address} {Oxford},\ \bibinfo {year}
  {2012})\BibitemShut {NoStop}%
\bibitem [{\citenamefont {Kanazawa}\ \emph
  {et~al.}(2015{\natexlab{a}})\citenamefont {Kanazawa}, \citenamefont {Sano},
  \citenamefont {Sagawa},\ and\ \citenamefont {Hayakawa}}]{KSSH2015I}%
  \BibitemOpen
  \bibfield  {author} {\bibinfo {author} {\bibfnamefont {K.}~\bibnamefont
  {Kanazawa}}, \bibinfo {author} {\bibfnamefont {T.~G.}\ \bibnamefont {Sano}},
  \bibinfo {author} {\bibfnamefont {T.}~\bibnamefont {Sagawa}}, \ and\ \bibinfo
  {author} {\bibfnamefont {H.}~\bibnamefont {Hayakawa}},\ }\href {\doibase
  10.1103/PhysRevLett.114.090601} {\bibfield  {journal} {\bibinfo  {journal}
  {Phys. Rev. Lett.}\ }\textbf {\bibinfo {volume} {114}},\ \bibinfo {pages}
  {090601} (\bibinfo {year} {2015}{\natexlab{a}})}\BibitemShut {NoStop}%
\bibitem [{\citenamefont {Kanazawa}\ \emph
  {et~al.}(2015{\natexlab{b}})\citenamefont {Kanazawa}, \citenamefont {Sano},
  \citenamefont {Sagawa},\ and\ \citenamefont {Hayakawa}}]{KSSH2015II}%
  \BibitemOpen
  \bibfield  {author} {\bibinfo {author} {\bibfnamefont {K.}~\bibnamefont
  {Kanazawa}}, \bibinfo {author} {\bibfnamefont {T.~G.}\ \bibnamefont {Sano}},
  \bibinfo {author} {\bibfnamefont {T.}~\bibnamefont {Sagawa}}, \ and\ \bibinfo
  {author} {\bibfnamefont {H.}~\bibnamefont {Hayakawa}},\ }\href {\doibase
  10.1007/s10955-015-1286-x} {\bibfield  {journal} {\bibinfo  {journal}
  {Journal of Statistical Physics}\ }\textbf {\bibinfo {volume} {160}},\
  \bibinfo {pages} {1294} (\bibinfo {year} {2015}{\natexlab{b}})}\BibitemShut
  {NoStop}%
\bibitem [{\citenamefont {Morgado}\ and\ \citenamefont
  {Queir\'os}(2016)}]{MorgadoQ2016}%
  \BibitemOpen
  \bibfield  {author} {\bibinfo {author} {\bibfnamefont {W.~A.~M.}\
  \bibnamefont {Morgado}}\ and\ \bibinfo {author} {\bibfnamefont {S.~M.~D.}\
  \bibnamefont {Queir\'os}},\ }\href {\doibase 10.1103/PhysRevE.93.012121}
  {\bibfield  {journal} {\bibinfo  {journal} {Phys. Rev. E}\ }\textbf {\bibinfo
  {volume} {93}},\ \bibinfo {pages} {012121} (\bibinfo {year}
  {2016})}\BibitemShut {NoStop}%
\bibitem [{\citenamefont {Belousov}\ \emph
  {et~al.}(2016{\natexlab{a}})\citenamefont {Belousov}, \citenamefont {Cohen},\
  and\ \citenamefont {Rondoni}}]{PRE2016II}%
  \BibitemOpen
  \bibfield  {author} {\bibinfo {author} {\bibfnamefont {R.}~\bibnamefont
  {Belousov}}, \bibinfo {author} {\bibfnamefont {E.~G.~D.}\ \bibnamefont
  {Cohen}}, \ and\ \bibinfo {author} {\bibfnamefont {L.}~\bibnamefont
  {Rondoni}},\ }\href {\doibase 10.1103/PhysRevE.94.032127} {\bibfield
  {journal} {\bibinfo  {journal} {Phys. Rev. E}\ }\textbf {\bibinfo {volume}
  {94}},\ \bibinfo {pages} {032127} (\bibinfo {year}
  {2016}{\natexlab{a}})}\BibitemShut {NoStop}%
\bibitem [{\citenamefont {Belousov}\ \emph
  {et~al.}(2016{\natexlab{b}})\citenamefont {Belousov}, \citenamefont {Cohen},
  \citenamefont {Wong}, \citenamefont {Goree},\ and\ \citenamefont
  {Feng}}]{PRE2016}%
  \BibitemOpen
  \bibfield  {author} {\bibinfo {author} {\bibfnamefont {R.}~\bibnamefont
  {Belousov}}, \bibinfo {author} {\bibfnamefont {E.~G.~D.}\ \bibnamefont
  {Cohen}}, \bibinfo {author} {\bibfnamefont {C.-S.}\ \bibnamefont {Wong}},
  \bibinfo {author} {\bibfnamefont {J.~A.}\ \bibnamefont {Goree}}, \ and\
  \bibinfo {author} {\bibfnamefont {Y.}~\bibnamefont {Feng}},\ }\href {\doibase
  10.1103/PhysRevE.93.042125} {\bibfield  {journal} {\bibinfo  {journal} {Phys.
  Rev. E}\ }\textbf {\bibinfo {volume} {93}},\ \bibinfo {pages} {042125}
  (\bibinfo {year} {2016}{\natexlab{b}})}\BibitemShut {NoStop}%
\bibitem [{\citenamefont {Bonella}\ \emph {et~al.}(2014)\citenamefont
  {Bonella}, \citenamefont {Ciccotti},\ and\ \citenamefont
  {Rondoni}}]{Bonella2014}%
  \BibitemOpen
  \bibfield  {author} {\bibinfo {author} {\bibfnamefont {S.}~\bibnamefont
  {Bonella}}, \bibinfo {author} {\bibfnamefont {G.}~\bibnamefont {Ciccotti}}, \
  and\ \bibinfo {author} {\bibfnamefont {L.}~\bibnamefont {Rondoni}},\ }\href
  {http://stacks.iop.org/0295-5075/108/i=6/a=60004} {\bibfield  {journal}
  {\bibinfo  {journal} {EPL (Europhysics Letters)}\ }\textbf {\bibinfo {volume}
  {108}},\ \bibinfo {pages} {60004} (\bibinfo {year} {2014})}\BibitemShut
  {NoStop}%
\bibitem [{\citenamefont {Dal~Cengio}\ and\ \citenamefont
  {Rondoni}(2016)}]{dalCengio2016}%
  \BibitemOpen
  \bibfield  {author} {\bibinfo {author} {\bibfnamefont {S.}~\bibnamefont
  {Dal~Cengio}}\ and\ \bibinfo {author} {\bibfnamefont {L.}~\bibnamefont
  {Rondoni}},\ }\href@noop {} {\bibfield  {journal} {\bibinfo  {journal}
  {Symmetry}\ }\textbf {\bibinfo {volume} {8}},\ \bibinfo {pages} {73}
  (\bibinfo {year} {2016})}\BibitemShut {NoStop}%
\bibitem [{\citenamefont {Onsager}(1931)}]{OnsagerI}%
  \BibitemOpen
  \bibfield  {author} {\bibinfo {author} {\bibfnamefont {L.}~\bibnamefont
  {Onsager}},\ }\href {\doibase 10.1103/PhysRev.37.405} {\bibfield  {journal}
  {\bibinfo  {journal} {Phys. Rev.}\ }\textbf {\bibinfo {volume} {37}},\
  \bibinfo {pages} {405} (\bibinfo {year} {1931})}\BibitemShut {NoStop}%
\bibitem [{\citenamefont {Haile}(1992)}]{Haile}%
  \BibitemOpen
  \bibfield  {author} {\bibinfo {author} {\bibfnamefont {J.}~\bibnamefont
  {Haile}},\ }\href@noop {} {\emph {\bibinfo {title} {Molecular dynamics
  simulation: elementary methods}}}\ (\bibinfo  {publisher} {Wiley},\ \bibinfo
  {address} {New York},\ \bibinfo {year} {1992})\BibitemShut {NoStop}%
\bibitem [{\citenamefont {Chandrasekhar}(1943)}]{Chandrasekhar}%
  \BibitemOpen
  \bibfield  {author} {\bibinfo {author} {\bibfnamefont {S.}~\bibnamefont
  {Chandrasekhar}},\ }\href {\doibase 10.1103/RevModPhys.15.1} {\bibfield
  {journal} {\bibinfo  {journal} {Rev. Mod. Phys.}\ }\textbf {\bibinfo {volume}
  {15}},\ \bibinfo {pages} {1} (\bibinfo {year} {1943})}\BibitemShut {NoStop}%
\bibitem [{\citenamefont {Weeks}\ \emph {et~al.}(1971)\citenamefont {Weeks},
  \citenamefont {Chandler},\ and\ \citenamefont {Andersen}}]{WCA}%
  \BibitemOpen
  \bibfield  {author} {\bibinfo {author} {\bibfnamefont {J.~D.}\ \bibnamefont
  {Weeks}}, \bibinfo {author} {\bibfnamefont {D.}~\bibnamefont {Chandler}}, \
  and\ \bibinfo {author} {\bibfnamefont {H.~C.}\ \bibnamefont {Andersen}},\
  }\href@noop {} {\ \textbf {\bibinfo {volume} {54}},\ \bibinfo {pages} {5237}
  (\bibinfo {year} {1971})}\BibitemShut {NoStop}%
\bibitem [{\citenamefont {Evans}\ and\ \citenamefont
  {Morriss}(2007)}]{EvansMorriss}%
  \BibitemOpen
  \bibfield  {author} {\bibinfo {author} {\bibfnamefont {D.~J.}\ \bibnamefont
  {Evans}}\ and\ \bibinfo {author} {\bibfnamefont {G.~P.}\ \bibnamefont
  {Morriss}},\ }\href@noop {} {\emph {\bibinfo {title} {{S}tatistical
  {M}echanics of {N}onequilibrium {L}iquids}}},\ \bibinfo {edition} {2nd}\ ed.\
  (\bibinfo  {publisher} {ANUE Press},\ \bibinfo {address} {Canberra},\
  \bibinfo {year} {2007})\BibitemShut {NoStop}%
\bibitem [{\citenamefont {Lagar'kov}\ and\ \citenamefont
  {Sergeev}(1978)}]{LagarkovSergeev}%
  \BibitemOpen
  \bibfield  {author} {\bibinfo {author} {\bibfnamefont {A.~N.}\ \bibnamefont
  {Lagar'kov}}\ and\ \bibinfo {author} {\bibfnamefont {V.~M.}\ \bibnamefont
  {Sergeev}},\ }\href {http://stacks.iop.org/0038-5670/21/i=7/a=R02} {\bibfield
   {journal} {\bibinfo  {journal} {Soviet Physics Uspekhi}\ }\textbf {\bibinfo
  {volume} {21}},\ \bibinfo {pages} {566} (\bibinfo {year} {1978})}\BibitemShut
  {NoStop}%
\bibitem [{\citenamefont {Kenney}\ and\ \citenamefont
  {Keeping}(1951)}]{KenneyKeepingII}%
  \BibitemOpen
  \bibfield  {author} {\bibinfo {author} {\bibfnamefont {J.}~\bibnamefont
  {Kenney}}\ and\ \bibinfo {author} {\bibfnamefont {E.}~\bibnamefont
  {Keeping}},\ }\href@noop {} {\emph {\bibinfo {title} {{M}athematics of
  {S}tatistics}}},\ \bibinfo {edition} {2nd}\ ed.,\ Vol.~\bibinfo {volume} {2}\
  (\bibinfo  {publisher} {D. Van Nostrand Company},\ \bibinfo {address} {New
  York},\ \bibinfo {year} {1951})\BibitemShut {NoStop}%
\bibitem [{\citenamefont {Adamo}\ \emph {et~al.}(2014)\citenamefont {Adamo},
  \citenamefont {Belousov},\ and\ \citenamefont
  {Rondoni}}]{AdamoBelousovRondoni}%
  \BibitemOpen
  \bibfield  {author} {\bibinfo {author} {\bibfnamefont {P.}~\bibnamefont
  {Adamo}}, \bibinfo {author} {\bibfnamefont {R.}~\bibnamefont {Belousov}}, \
  and\ \bibinfo {author} {\bibfnamefont {L.}~\bibnamefont {Rondoni}},\
  }\enquote {\bibinfo {title} {Fluctuation-dissipation and fluctuation
  relations: From equilibrium to nonequilibrium and back},}\ in\ \href
  {\doibase 10.1007/978-3-642-54251-0_4} {\emph {\bibinfo {booktitle} {Large
  Deviations in Physics: The Legacy of the Law of Large Numbers}}},\ \bibinfo
  {editor} {edited by\ \bibinfo {editor} {\bibfnamefont {A.}~\bibnamefont
  {Vulpiani}}, \bibinfo {editor} {\bibfnamefont {F.}~\bibnamefont {Cecconi}},
  \bibinfo {editor} {\bibfnamefont {M.}~\bibnamefont {Cencini}}, \bibinfo
  {editor} {\bibfnamefont {A.}~\bibnamefont {Puglisi}}, \ and\ \bibinfo
  {editor} {\bibfnamefont {D.}~\bibnamefont {Vergni}}}\ (\bibinfo  {publisher}
  {Springer Berlin Heidelberg},\ \bibinfo {address} {Berlin, Heidelberg},\
  \bibinfo {year} {2014})\ pp.\ \bibinfo {pages} {93--133}\BibitemShut
  {NoStop}%
\bibitem [{\citenamefont {Boon}\ and\ \citenamefont {Yip}(1980)}]{BoonYip}%
  \BibitemOpen
  \bibfield  {author} {\bibinfo {author} {\bibfnamefont {J.~P.}\ \bibnamefont
  {Boon}}\ and\ \bibinfo {author} {\bibfnamefont {S.}~\bibnamefont {Yip}},\
  }\href@noop {} {\emph {\bibinfo {title} {Molecular hydrodynamics}}}\
  (\bibinfo  {publisher} {McGraw-Hill},\ \bibinfo {address} {New York},\
  \bibinfo {year} {1980})\BibitemShut {NoStop}%
\bibitem [{\citenamefont {Steutel}\ \emph {et~al.}(1979)\citenamefont
  {Steutel}, \citenamefont {Kent}, \citenamefont {Bondesson},\ and\
  \citenamefont {{Barndorff-Nielsen}}}]{Steutel}%
  \BibitemOpen
  \bibfield  {author} {\bibinfo {author} {\bibfnamefont {F.}~\bibnamefont
  {Steutel}}, \bibinfo {author} {\bibfnamefont {J.}~\bibnamefont {Kent}},
  \bibinfo {author} {\bibfnamefont {L.}~\bibnamefont {Bondesson}}, \ and\
  \bibinfo {author} {\bibfnamefont {O.}~\bibnamefont {{Barndorff-Nielsen}}},\
  }\href {http://www.jstor.org/stable/4615732} {\bibfield  {journal} {\bibinfo
  {journal} {Scandinavian Journal of Statistics}\ }\textbf {\bibinfo {volume}
  {6}},\ \bibinfo {pages} {57} (\bibinfo {year} {1979})}\BibitemShut {NoStop}%
\bibitem [{\citenamefont {Frenkel}\ and\ \citenamefont
  {Smit}(2002)}]{FrenkelSmit}%
  \BibitemOpen
  \bibfield  {author} {\bibinfo {author} {\bibfnamefont {D.}~\bibnamefont
  {Frenkel}}\ and\ \bibinfo {author} {\bibfnamefont {B.}~\bibnamefont {Smit}},\
  }\href@noop {} {\emph {\bibinfo {title} {{U}nderstanding of {M}olecular
  {D}ynamics {S}imulations: {F}rom {A}lgorithms to {A}pplications}}}\ (\bibinfo
   {publisher} {Academic Press},\ \bibinfo {address} {San Diego},\ \bibinfo
  {year} {2002})\BibitemShut {NoStop}%
\bibitem [{\citenamefont {Martyna}\ \emph {et~al.}(1996)\citenamefont
  {Martyna}, \citenamefont {Tuckerman}, \citenamefont {Tobias},\ and\
  \citenamefont {Klein}}]{MartynaTuckermanII}%
  \BibitemOpen
  \bibfield  {author} {\bibinfo {author} {\bibfnamefont {G.~J.}\ \bibnamefont
  {Martyna}}, \bibinfo {author} {\bibfnamefont {M.~E.}\ \bibnamefont
  {Tuckerman}}, \bibinfo {author} {\bibfnamefont {D.~J.}\ \bibnamefont
  {Tobias}}, \ and\ \bibinfo {author} {\bibfnamefont {M.~L.}\ \bibnamefont
  {Klein}},\ }\href {\doibase 10.1080/00268979600100761} {\bibfield  {journal}
  {\bibinfo  {journal} {Molecular Physics}\ }\textbf {\bibinfo {volume} {87}},\
  \bibinfo {pages} {1117} (\bibinfo {year} {1996})},\ \Eprint
  {http://arxiv.org/abs/http://dx.doi.org/10.1080/00268979600100761}
  {http://dx.doi.org/10.1080/00268979600100761} \BibitemShut {NoStop}%
\bibitem [{\citenamefont {Martyna}\ and\ \citenamefont
  {Tuckerman}(1995)}]{MartynaTuckermanI}%
  \BibitemOpen
  \bibfield  {author} {\bibinfo {author} {\bibfnamefont {G.~J.}\ \bibnamefont
  {Martyna}}\ and\ \bibinfo {author} {\bibfnamefont {M.~E.}\ \bibnamefont
  {Tuckerman}},\ }\href@noop {} {\bibfield  {journal} {\bibinfo  {journal} {The
  Journal of Chemical Physics}\ }\textbf {\bibinfo {volume} {102}} (\bibinfo
  {year} {1995})}\BibitemShut {NoStop}%
\end{thebibliography}%

\end{document}